\documentclass[12pt]{article}
\usepackage{amsmath}
\usepackage{amssymb}
\usepackage{amsthm}
\usepackage{amscd}
\usepackage{amsfonts}
\usepackage{graphicx}%
\newtheorem{theorem}{Theorem}[section]

\newtheorem{definition}{Definition}[section]

\newtheorem{proposition}{Proposition}[section]
\newtheorem{remark}{Remark}[section]
\begin{document}
\title{\bf Encryption Based on Conference Matrix}

\author {Shipra Kumari\thanks {E-mail: shipracuj@gmail.com} ~and~ Hrishikesh Mahato\thanks{Corresponding author: E-mail: hrishikesh.mahato@cuj.ac.in }\\
	\small Department of Mathematics,
	\small Central University of Jharkhand,
	\small Ranchi-835205, India}
\date{}
\maketitle \setlength{\parskip}{.11 in}
\setlength{\baselineskip}{15pt}
\maketitle
\begin{abstract}
In this article, an encryption scheme based on $(-1,1)$ conference matrix has been developed. The decryption key comprising of fixed number of positive integers with prime power yields the high level security of message. Some popular attacks has been discussed in the context of cryptoanalysis and  observed that it is robust against the popular cipher attack and the security of the information does not compromise.		
\end{abstract}
\hrule

\section{Introduction:}
Cryptography is the study of techniques of secured communications i.e. a study of techniques which ensure that communicated information cannot be understood by anyone except the intended receiver.\\
% There are several cryptographic algorithms to encrypt and decrypt the plain text for secure communication.\\ 
%The sender of encrypted text shares the decoding technique with intended receiver.\\
In $1929$ Lester S. Hill introduced the Hill cipher in which an invertible matrix is used as a private key  and inverse of that matrix is used to decrypt the message.\\
In this article we propose a private symmetric key encryption scheme based on conference matrices of order $n$ such that $n= p^r(p+2)^{r'}$ : where $r=1, r'=1$ or $r\geq 1, r'=0$, and $p, \ p+2$  are any odd primes. 
\\
%This is the basic idea for the proposed encryption scheme.
However in this encryption scheme a modular based message has been encrypted by a conference matrix of order $n$ and decrypted  by its transpose, the identity matrix $I_{n}$ and $J_{n}$ (a square matrix of order $n$  with all entries $1$). It is not easy to find the conference matrix of order $n$ unless such $n$ is given. Subsequently it may raise the difficulties for the intruder by involvement of a positive integer modular base and order of conference matrix $n$ formed by a large prime.  \\  
C. Koukouvinos and D.E. Simos \cite{koukouencrypt} also developed an encryption scheme using circulant Hadamard core  in which the first row of the Hadamard core required to be transmitted as a private key to the intended receiver.\\
 But this scheme requires transmission of numbers $(p,r,r',d,q)$ as a private  key, where $n=p^r(p
 +2)^{r'}$:  either $r= 1, r'= 1$ or $r \geq 1, r'=0$
  is the size of conference matrix, $d$ is any constant and $q$ a positive integer modular base. In addition to these  numbers the primitive polynomial $P(n)$ of $GF(p^r)$ which has been used to construct the conference matrix of order $n$ is required to transmit in case of $n=p^r, r>1$ \cite{Pal}.  The involvement of the transpose of conference matrix and two standard matrices make easy the construction of the decryption key for intended receiver.\\
 % But it is quite difficult to guess the private key for intruder for large value of $n$ as difficulty of prime factorization of $n$. \\
 The main goal of the proposed technique includes the following:
\begin{enumerate}
	\item Require the private key which is shared by the sender and receiver only once.
	\item Easy transmission of private key
	\item  Computation of encryption and decryption are fast.
	\item Difficult to guess the key for intruder.
	\item Robust to cryptographic attack.
\end{enumerate}
The structure of this paper is as follows. 
In section $2$, the required definitions and information has been given as preliminaries. In section $3$ we have discussed the encryption/decryption algorithms and its mechanism which ensure that the encrypted message is determined uniquely. Furthermore in section $4$ we have done security analysis of the scheme with different cryptographic attack. In section $5$ we have given an example
 and finally in section $6$ we have concluded the paper explaining its limitations and benefits.
\section{Preliminaries}
In this section some basic terminologies are defined which has been used to design the cryptographic algorithm.

\begin{definition}\textbf{Encryption scheme} \cite{koukouencrypt}\\ 	
An Encryption is the process in which we encode a message or information in such a manner so that only intended person can access it. There are three sets in the encryption scheme: a message set or plaintext $M$, a ciphertext (encrypted message) $C$, and a key set $K$  together with the following three algorithms.
\begin{enumerate}
	\item  A key set $K$ which generates the valid encryption key $k \in K$ and a valid key $k^{-1} \in K$ to decrypt the message.
	\item An encryption algorithm in which message $m \in M$ and key $k \in K$ together produce an element $c \in C$ which is defined as $c=E_{k}(m)$.
	\item A decryption algorithm in which an element $c \in C$ with decryption key $k^{-1}\in K$  return back an element of message $m \in M$ with $m=D_{k^{-1}}(c)$. 
\end{enumerate}
Note that $D_{k^{-1}}(E_{k}(m))=m$.
\end{definition}
\begin{definition} \textbf{$\mathcal{O}$-notation}  \\
This notation is used to describe the complexity or performance of an algorithm. Basically "big $\mathcal{O}$"  defines an upper bound of an algorithm.	Formally, If $f(n)$ and $g(n)$ are two functions, we denote $\mathcal{O}(g(n))$ the set of functions and defined as\\
$\mathcal{O}(g(n))=\{f(n)$ : there exist positive constant $c$ and $n_{0}$ such that $0\leq f(n) \leq cg(n)$ for all $n\geq n_{0}$\}
	\end{definition}

\begin{definition} \textbf{Conference matrix} \label{def:conference}\\
	A $(-1,1)$ square matrix $A$ of order $n$ is known as conference matrix if
	\\ $$A^TA=AA^T=(n+1)I_{n}-J_{n}$$
		
\end{definition}

\subsection{Existence of Conference matrix}\label{construct:paley}
There are some methods to construct Hadamard matrices of order $n+1$ which ensure the existence of conference matrix of order $n$.
\begin{enumerate} 
	\item If $n=p^r$, where $n \equiv 3(mod \ 4)$,  $p$ is an odd prime and $r$ is a positive integer. Then using Paley construction we can get conference matrix of order $n$ \cite{Pal}.
	
	\item $n=p(p+2)$, where $p$ and $p+2$ are prime then using difference set there exist a conference matrix of order $n$. \cite{mar}
	%	\item $p=2^t-1$, where $t$ is a positive integer.\cite{singer}.
\end{enumerate}

The following remarks determine $A^T$ directly without formation of $A$
\begin{remark}\label{construct:p^r}
	Let $n=p^r$, where $p$ is an odd prime, $r \geq 1$ and $n\equiv 3(mod \ 4)$ \cite{Pal}. If the leading elements of conference matrix $A$ are  $a_{0}, a_{1}, \cdots, a_{n-1}$ 
	and 	\begin{equation}
	A:(i,j)=\chi(a_{j}-a_{i})
	\end{equation}
	then
	\begin{equation}
	A^T:(i,j)=\chi(a_{i}-a_{j})
	\end{equation}
	and vice-versa.\\
	However for $r=1$ the leading elements of $A$ are $0,1,2, \cdots n-1$ and $A$ is circulant. For $r >1$ the leading elements of $A$ are $\lambda^{0}, \lambda^{1}, \cdots, \lambda^{n-1}$ where $\lambda $ is a  root of primitive polynomial $P(n)$ of $GF(p^r)$.
\end{remark}
\begin{remark}\label{construct:pq}
	Let $n=p(p+2)$, where $p,p+2$ are prime. The  leading elements of $A$ are $0,1, \cdots n-1$ and the first row of the circulant matrix $A=(a_{i,j})$ is determined  by difference set \cite{mar} then 
	\begin{equation}
	A^T= (a_{i,(n-j)(mod \ n)})
	\end{equation}
	
\end{remark}
It can be noted that if we fix with prescribed leading elements of $A$, its transpose $A^T$ may be obtained directly.
\section{Results}
\subsection{Design of Cryptographic Algorithm}

Let there are $q$ distinct characters in the language in which the message or information is written. We convert the message to be transmitted into its corresponding numeric plain text (ASCII code) in modulo $q$. In order to block cipher we divide the plain text into blocks of each size $n$ and each block represented as a column vector. We add  "space" in the last block to make it of size $n$ if needed. \\
 The encrypted message to be transmitted over a communication channel of a column vector $M$  is
\begin{equation}\label{encrypt:C}
C \equiv(AM+de_{n})(mod \ q)
\end{equation}
where $d$ is any constant, $e_{n} =(1,1,\cdots,1)^T$ , $A$ is a conference matrix of order $n$ and $q$ is a positive integer modular base with $gcd(n+1, q)=1$. In context of construction of $A$, the leading elements of $A$ are $0,1,2, \cdots, n-1$ for $n=p$ or $n=p(p+2)$ and those of $A$ are $\lambda^{0}, \lambda^{1}, \cdots , \lambda^{n-1}$ for $n=p^r, r>1$ where $\lambda $ is a primitive root of $GF(p^r)$.
According to Hill cipher it requires $A^{-1}$ to decrypt the message. However in this scheme $C-de_{n}$ is pre-multiplied by $A^T$ by the intended receiver to disposed off the calculation of $A^{-1}$. Now to get the original message receiver has to decrypt the message using the transformation
\begin{equation} \label{decrypt:M}
M\equiv  (A^TA)^{-1}A^T(C-de_{n})(mod \ q)
\end{equation}
%The matrix key $A$ which is used in encryption scheme is the most vital component.
 
which requires to calculate the  $(A^TA)^{-1}$ with modular base $q$ and $A^T$. In general it is quite difficult to calculate $(A^TA)^{-1}$ for large value of $n$. 
  But, since $A$ is a conference matrix of order $n$, we have 
  \begin{equation*}
  A^TA= (n+1)I_{n}-J_{n}
  \end{equation*}
 so
  \begin{equation}
  (A^TA)^{-1}= \frac{1}{n+1}(I_{n}+J_{n})
  \end{equation} 
% \textbf{Note:} It is noted that gcd$(p+1,q)=1$. otherwise $(A^TA)^{-1}$ doesn't exist.
  \\
  and equation \eqref{decrypt:M} reduced to 
  \begin{equation}
  M \equiv \frac{1}{n+1}(I_{n}+J_{n})A^T(C-de_{n})(mod \ q)
  \end{equation}
 which is quite easy to form the decryption key for intended receiver as $n$ is known.
  %In this encryption scheme $(-1,1)$ conference matrix is used as a key matrix.
   %Also to raise the difficulties for intruder modular base $q$  is used in encryption and decryption.\\
%Here order of a conference matrix $A$ is $p$, a prime, of the from of $4k+3$. By Paley construction conference matrix $A$ can be constructed.\\
%\textbf{Note:} If $A$ is a circulant matrix with first row $[a_{(1,1)},a_{(1,2)}, \cdots , a_{(1,p-1)}, a_{(1,p)}]$ then first row of $A^T=[a_{(1,i)},]$, are \begin{equation}a_{(1,i)}=a_{(1,(p-(i-2))(mod \ p))}
%\end{equation} and $1 \leq i \leq p$

The cryptographic algorithm for encryption is given by\\
\textbf{Encryption Algorithm}\\
\textbf{Function EncrAlg(msg)}\\
Require msg to encrypt\\
  select $(p, r, r',d,q, P(p^r))$\\
\begin{equation*}
\left.\begin{aligned}
k & \leftarrow (p, r, r', d, q) \hspace{1.1cm} \ if \ r=1 \ and \ r' = 1 \\ 
k & \leftarrow (p, r, r', d, q, P(p^r)) \ if \ r \geq 1 \ and \  r'=0  \hspace{1cm} 
\end{aligned}\right\} Form \ private \ key
\end{equation*}
Transmit$(k)$ \hspace{7cm} Transmit the securely the private key\\
$M \leftarrow $ convert(msg) \hspace{5.9cm} Convert original msg\\
calculate $n=p^r(p+2)^{r'}$\\
$C\leftarrow (AM+de_{n})(mod \ q)$ \hspace{5cm} Encrypted msg is $C$\\
Return(Transmit(C))\\
\textbf{End Function}

In order to fulfill the objectives of the cryptography the encrypted message $C$ has to be decrypted uniquely.
\begin{theorem}
	If $C $ is the encrypted message which is transmitted with the encryption algorithm then the decrypted message $D \equiv (I_{n}+J_{n})A^T(C-de_{n})t(mod \ q)$ is uniquely determined and is equal to $M$, where $t$ is solution of $(n+1)x \equiv 1(mod \ q)$.
\end{theorem}
\begin{proof}
	Since $gcd(n+1,q)=1$ so $(n+1)x \equiv 1(mod \ q)$ has unique solution $t$.
	As $C$ is an encrypted message with respect to the encryption algorithm \eqref{encrypt:C}. So\begin{align*}
	C &\equiv(AM+de_{n})(mod \ q)\\
	\Rightarrow C-de_{n} &\equiv AM (mod \ q)
	\end{align*} 	
Since,\begin{align*}
D & \equiv (I_{n}+J_{n})A^T(C-de_{n})t(mod \ q)\\
&\equiv(I_{n}+J_{n})A^TAMt(mod \ q) & Since \ A \ is \ a \ conference \ matrix\\
 &\equiv(I_{n}+J_{n})((n+1)I_{n}-J_{n})Mt(mod \ q)\\
&\equiv \{(n+1)I_{n}-J_{n}+(n+1)J_{n}-nJ_{n}\}Mt(mod \ q)\\
&\equiv(n+1)I_{n}Mt(mod \ q)\\
&\equiv M(n+1)t(mod \ q)\\
&\equiv M (mod \ q)
\end{align*}
	
	 So, $D=M$ i.e. message is uniquely decrypted.
	\end{proof}
	\textbf{Decryption Algorithm}
	
	 \textbf{Function DecrAlg($C$)}\\
	 Require received ciphertext $C$\\
	 received $(p, r, r',d,q, P(p^r))$\\
	 \begin{equation*}
	 \left.\begin{aligned}
	 k \leftarrow (p, r, r', q)  \hspace{1cm} \ \  if \ r=1 \ and  \ r' = 1 \\ 
	 k \leftarrow (p, r, r', q, P(p^r)) \  \ if \ r \geq 1 \ and \ r'=0
	 \end{aligned}\right\} set\ private \ key
	 \end{equation*}
	 calculate $n=p^r(p+2)^{r'}$\\
	 	 $M\leftarrow (k(C-de_{n}))(mod \ q)$ \hspace{3.7cm} Decrypt ciphertext\\
	 msg $\leftarrow$ convert($M$)\\
	 Return(msg)\\
	 \textbf{End Function}
	
%An encryption scheme must be robust against the attacks and doesn't compromise the security of information. %Strength of an encryption scheme is determined by the computational power needed to break it.\\

\subsection{Analysis of time complexity of algorithm}
In the above mentioned encryption scheme sender transmit  numbers $(p, r, r',d,q)$ and $P(n)$ (in either case) as a private key. To get the original message intended receiver has to use the transformation 
\begin{equation*}
M\equiv (I_{n}+J_{n})A^T(C-de_{n})t(mod \ q)
\end{equation*} where $I_{n}$ is an identity matrix of order $n=p^r(p+2)^{r'}$ and $J_{n}$ is a square matrix of order $n$ with all entries $1$.
It means intended receiver has to find out $A^T$ and $(n+1)^{-1} (mod \ q)$ i.e.$t$ only. $A^T$ can be obtained by using remark \eqref{construct:p^r} and \eqref{construct:pq}.
%Since $n$ is a prime such that $n \equiv 3(mod \ 4)$ so $A^T$ can be construct using Quadratic residues of $\mathbb{Z}_{n}$ and the first row of circulant matrix $A^T$  is defined as 
%$a_{1i}=\chi((n-i)mod \ n)$ where $i=0,1,2, \cdots, n-1$ and $\chi$ is extended quadratic character. There are $\frac{n-1}{2}$ quadratic residues in $\mathbb{Z}_{n}$. The time complexity of finding quadratic residues is  $\mathcal{O}(\frac{n-1}{2})$.
 Thus the time complexity to find $A^T$ is $\mathcal{O}(n)$.
 In this scheme there is no need to calculate the inverse of $A^TA$. Since $gcd(n+1, q)=1$ so, $t$ may be obtained  using Euclidean algorithm and its time complexity is $\mathcal{O}(log \ n)$. Thus we see that the time complexity of this algorithm depends on matrix multiplications of the matrices $(I_{n}+J_{n}), A^T$ and the column vector $C-de_{n}$ and is $\mathcal{O}(n^3)$.

\section{Security of the method}	
	
	In cryptography, the main aim is to protect the information about the key, plain text and cypher text from the intruder. But intruder always tries to attack a cipher or cryptographic system so that they can get a lead to break it fully or only partially. \\  The types of main attacks are follows:
	\begin{itemize}
		\item Brute force attack
		\item Known Plain text attack
		\item Ciphertext-only attack
			
	\end{itemize} 
 
\subsection{Cryptanalysis of Brute force attack}
In Brute force attack, the intruder tries all possible character combination to find the keys and checks which one of them returns the plaintext. \\
In order to break above discussed encryption scheme using brute force attack, intruder has to find the key matrix $A$ and $t$. It is difficult to find $A$ and $t$ unless $n, q$ and $P(n)$ in either case are known. Since the set of numbers  $(p, r, r',d,q)$ and $P(n)$ is a private key comprising a large prime $p$ and primality testing is NP. However there is an algorithm of primality testing for a given integer is in $P$ \cite{agrawal}.  In  either case to obtain the suitable primitive polynomial of $GF(p^r)$ increases the difficulty level.\\ So it will be  difficult to guess the size of $A$ exactly. Suppose any one could guess the order of key matrix $A$. Since matrix $A$ consists $(-1,1)$ only, so the size of the key space, $ K(A)$, is $|K(A)|=2^{n^2}$. Thus the computational complexity to find $A$ is $\mathcal{O}(2^{n^2})$.\\
 Its complexity increases exponentially. Thus it seems that the above discussed encryption scheme is robust against the Brute force attack.

\subsection{Cryptanalysis of known Plaintext attack}
	The known plaintext attack is the one where intruder has an access to the quantity of plaintext as well as its corresponding cipher text. In this type of attack the main goal is to guess the private key or to develop an algorithm so that they can decrypt any further message.
	In the above discussed encryption scheme,
	 we have
	 \begin{equation*} C\equiv (AM+de_{n})(mod \ q)
	 \end{equation*}
So, basially to find the encryption scheme they have to solve the $n-$ dimensional non-homogeneous system of linear equation which is very difficult. 	

\begin{proposition} \cite{koukouencrypt}
	All encryption scheme using Hadamard matrices (conference matrix) with circulant cores are secure against known-plain text attacks under the assumption that the adversary has knowledge of less than $n$ messages of length $n$ of the plain text and the corresponding cipher text. 
\end{proposition}
%\subsection{Cryptanalysis of chosen plaintext attack}
%The chosen plaintext attack is the one where the intruder choose arbitrary plaintext data and receive a ciphertext. Actually in this type of attack they try to create an algorithm to decrypt the message. 
%As, 
%\begin{equation*}
%C\equiv (AM+de_{n})(mod \ q)
%\end{equation*}, so to find the encryption key they have to solve the non-homogenous system of linear equation with $n$ variables.
%\begin{proposition} \cite{koukouencrypt}
%	All encryption scheme using Hadamard matrices with circulant cores are secure against chosen-plaintext attacks, as the scheme are secure against known-plaintext attacks. 
%\end{proposition}
%In known-plaintext attack, attacker has atleast one known sample of both plaintext and ciphertext, whereas in chosen-plaintext attack, attacker specify their own plaintext and encrypt that.\\

\subsection{Cryptanalysis of ciphertext- only attacks} 
The ciphertext-only attack is the one where intruder has access to the number of encrypted message. They have no idea about exact plain text and private key is. In this type of attack the main goal is to deduce the private key or plain text. Mainly they focus on finding the private key so that they can use that to decrypt the further encrypted message.\\
 So, to design the encryption algorithm it is particularly important to protect them against the cipher text only attack. As we can say this attack is the starting point of cryptanalyst.
 
When we use conference matrix in encryption scheme two same letters of the plain text  $M$ corresponds to different values of the encrypted text $C$. So, an attacker cannot observe the plain text or any information regarding the private key after seeing the encrypted message. 
\begin{proposition}\cite{koukouencrypt}
	All encryption scheme using Hadamard matrices are secure against ciphertext-only attack.
\end{proposition} 
%\subsection{Cryptanalysis of chosen-ciphertext attacks}
%The chosen ciphertext attack is the one where they can analyse any arbitrary chosen ciphertext with their corresponding plaintext. The goal of this attack is to get a private key or as much as possible to collect the information about the attacked encryption scheme.\\
%In the above disscussed encryption scheme it is not possible to guess the private key, as we have used conference matrix for encryption. Since any value of the cipher text is a combination of $n$ values of a plain text and one row of a key matrix $A_{n}$. So, if there are same values in the cipher text then it doesn't represent that plain text contains also same letters.
%\begin{proposition}\cite{koukouencrypt}
%	All encryption scheme using conference matrix (Hadamard matrices) are secure against chosen-cipher text and cipher text-only attack as well as known-plain text and chosen plain text attack.
%\end{proposition}
\section{Example}
Consider a message \textbf{ HELLO} which has to be transmit using the encryption scheme discussed above.\\ 
Message in ASCII code is
\begin{equation*}
 M= \left[ \begin{array}{ccccc}
 72 & 69 & 76 & 76 & 79
 \end{array}
 \right]^T
\end{equation*}\\
Suppose $p=19, r=1, r'=0$ so size of the conference matrix is $n=19$. \\
Since $M$ contains $5$ letters and $n=19$ so to make it equal "space" is added  $14$ times in $M$. The ASCII code of "space" is $32$. Thus 
\begin{equation*}
M= \left[
\begin{array}{ccccccccccccccccccc}
72 & 69  & 76 & 76 & 79 & 32 & 32 &  32 &  32 & 32 & 32 & 32 & 32 & 32 &  32 & 32 & 32 &  32 & 32\\
\end{array}
\right]^T
\end{equation*}
 So for modular base we can take $q= 81$ as $gcd(20,81)=1$
So, $n=19$, and suppose $d=2$. Thus encrypted message 
$C= (AM+de_{n})(mod \ 81)$.\\
 where the first row of the circulant conference matrix $A$ of order $19$   obtained by using construction defined in subsection \eqref{construct:paley} is  given by 
\begin{equation*}
A=\left[
\begin{array}{ccccccccccccccccccc} 
1& 1 &-1 & -1 & 1 & 1 & \ 1 & \ 1 & - 1 & \ 1 & - 1 & \ 1 & -1 & -1 & -1 & - 1 & \ 1 & \ 1 & -1\\
\end{array}
\right]
\end{equation*}
Therefore
\begin{equation*}
C= \left[
\begin{array}{ccccccccccccccccccc}
70 &	65 &	78 &	77 &	78 &	79 &	4 &	78 &	64 &	58 &	71 &	65 &	3 &	64 &	77 &	71 &	4 &	11 &	3\\
\end{array}
\right]^T
\end{equation*}

Generally $0-31$ and $127$ are not printable and it is indicated with "NA" \cite{online}. But here for our convenience we use 
\begin{align*}
0 & \rightarrow 0*\\
1 & \rightarrow 1 *\\
2 & \rightarrow 2*
\end{align*}
and so on. In this case intended receiver has to understand that when $n*$ is included in encrypted message its numeric value will be $n$, where $n$ is non printable character.
After converting the ASCII code of encrypted message  into its corresponding printable character
 \\ C= \textbf{F \ A \ N \ M \ N \ O \ 4* \ N \ @ \ : \ G \ A \ 3* \ @ \ M \ G \ 4* \ 11* \ 3*}.\\
Sender need to send the private key $(19, 1, 0, 2, 81)$ along with the encrypted message $C$.
Now intended receiver get plain text using transformation 
\begin{equation*}
M=(I_{n}+J_{n})A^T(C-de_{n})t (mod \ 81)
\end{equation*}
Matrices $I_{n}$ and $J_{n}$ are well known. $A^T$ and $t$ may be obtained with the help of remark \eqref{construct:p^r} and  Euclidean algorithm respectively. Thus receiver decrypt the message and get the column vector 
\begin{equation*}
M= \left[
\begin{array}{ccccccccccccccccccc}
	72 & 69  & 76 & 76 & 79 & 32 & 32 &  32 &  32 & 32 & 32 & 32 & 32 & 32 &  32 & 32 & 32 &  32 & 32\\
\end{array}
\right]^T
\end{equation*}

\section{Conclusion}
In this article we have developed an encryption scheme using conference matrix. The sender shares only the numbers $(p, r, r', d, q)$ and $P(n)$ (in either case) as a private key to the intended receiver. Private key comprising of limited numbers makes easy transmission.
The theoretical development of decryption key  makes easy to decrypt the massage for intended receiver.
 However for intruder it is very difficult to guess the private key as finding prime $p$ is an NP problem. Also obtaining the suitable primitive polynomial increases the difficulty level in case of involvement of prime power. %if $n=p^r$ then additional difficulty level will increase as primitive polynomial is used to construct $A^T$ it means intruder has to find out the primitive polynomial as well.
  It has been observed that the encryption scheme is robust against the cipher attack.

\end{document}